%% file: cm_main.tex
\documentclass[journal,comsoc]{IEEEtran}
\usepackage[T1]{fontenc}
\usepackage{cite}
\usepackage{amsmath}
\usepackage{url}
\usepackage{algorithm}
\usepackage{subfig}
\usepackage{graphicx}
\usepackage[noend]{algpseudocode}
\usepackage[flushleft]{threeparttable}
\input{mymath}
\begin{document}
\title{Developing Correlation Indices to Identify Coordinated Cyber-Attacks on Power Grids}
\author{Christian~Moya,~\IEEEmembership{Student~Member,~IEEE,}
        and~Jiankang~Wang,~\IEEEmembership{Member,~IEEE}
\thanks{C. Moya and J.K. Wang are with the Department of Electrical and Computer Engineering, The Ohio State University, Columbus, OH, 43210 USA e-mail: \{moyacalderon.1,wang.6536\}@osu.edu.}}

\maketitle
\begin{abstract}
Increasing reliance on Information and Communication Technology~(ICT) exposes the power grid to cyber-attacks. In particular, Coordinated Cyber-Attacks (CCAs) are considered highly threatening and difficult to defend against, because they (i) possess higher disruptiveness by integrating greater resources from multiple attack entities, and (ii) present heterogeneous traits in cyber-space and the physical grid by hitting multiple targets to achieve the attack goal. Thus, and as opposed to independent attacks, whose severity is limited by the power grid's redundancy, CCAs could inflict  disastrous consequences, such as blackouts. In this paper, we propose a method to develop Correlation Indices to defend against CCAs on static control applications. These proposed indices relate the targets of CCAs with attack goals on the power grid.  Compared to related works, the proposed indices present the benefits of deployment simplicity and are capable of detecting more sophisticated attacks, such as measurement attacks. We demonstrate our method using measurement attacks against Security Constrained Economic Dispatch. 
\end{abstract}
\IEEEpeerreviewmaketitle
\input{Introduction}
\input{Background}
\input{Modeling}
\input{AttackTemplate}
\input{DerivingCIs}
\input{CIsProperties}
\input{DefenseImplications}
\input{Applications}
\input{Experiments}

\section{Conclusion}
In this paper, we provided a method to derive Correlation Indices~(CIs) based on attack goals, which can be used to estimate attack consequences and identify critical substations during coordinated attacks. Compared to existing approaches, our method does not rely on case-by-case numerical simulations of control attacks with few attack goals. Instead, our method computes the CIs using an analytical attack template that can model more sophisticated attacks, such as measurement attacks. We modeled the attack template as a bilevel optimization program and derived Algorithm~\ref{alg:ImpactiveAttacks} to solve it. Algorithm~\ref{alg:ImpactiveAttacks} computes the CIs for any given attack goal. These CIs describe strongly correlated attacks, since the adversary reaches the goal by attacking the least number of target substations. We then used a set-theoretic approach to derive the CIs' properties. These properties suggest defense implications against coordinated attacks, including the best defense for a transmission line, the best defense against strongly correlated attacks, and the metric of defense effectiveness. Thus, our method to compute CIs and their properties presents the benefit of deployment simplicity but faces one limitation, namely the computational performance of Algorithm~\ref{alg:ImpactiveAttacks}. However, given that there is only few substations in the power grid, the computation performance is unlikely to be a problem. In our future work, we will use the CIs and their defense implications together with Intrusion Detection Systems to protect the grid against coordinated attacks.
\bibliographystyle{IEEEtran}
\bibliography{bib/cm_ref}
\section{Appendix}
\input{Proofs}
\end{document}

%% file: mymath.tex
\usepackage{amssymb,amsmath,amsfonts,amsthm,graphicx,wrapfig,etoolbox,nomencl,xpatch,mathtools,tcolorbox,empheq}
\usepackage{color}
\theoremstyle{plain}

\newtheorem{lemma}{Lemma}

\newtheorem{theorem}{Theorem}
\newtheorem{proposition}{Proposition}

\theoremstyle{definition}
\newtheorem{definition}{Definition}
\theoremstyle{remark}

\newtheorem{remark}{Remark}

\newcommand{\R}{{\mathbb R}}

\newcommand{\txt}[1]{\text{#1}} 

\def\eg{\emph{e.g., }}
\def\ie{\emph{i.e., }} 
\def\dist{0em}  
\newcommand{\grhl}[1]{\textcolor{gray}{#1}} 

\DeclareMathOperator*{\diag}{diag}

\newcommand{\Sa}{{S_{\alpha,j}}}                                            
\newcommand{\CSa}{{\mathcal{S}_\alpha (e,\tilde{\tau})}}        
\newcommand{\CI}{{S^*_{\alpha,j}}}                              
\newcommand{\CIb}{{S^*_{\beta,j}}} 
\newcommand{\CCI}{{\mathcal{S}^*_\alpha (e,\tilde{\tau})}}      
\newcommand{\Sp}{{S'_{\alpha,j}}}                              
\newcommand{\CCIb}{{\mathcal{S}^*_\beta (e,\tilde{\tau})}}      

\newcommand{\Db}{{D_{\beta,e}}}                                
\newcommand{\Ri}{{R(\tilde{\tau})}}                            
\newcommand{\Rb}{{R_{\beta,s_k}(\tilde{\tau})}}                
\newcommand{\DRb}{{\Delta R_{\beta,s_k}(\tilde{\tau})}}        

\newcommand{\eT}{{(e,\tilde{\tau})}} 

%% file: Introduction.tex
\section{Introduction}
The operation of today's power grid largely relies on automated control applications and Supervisory Control and Data Acquisition~(SCADA) systems. While control applications compute the commands to operate the power grid, SCADA serves  as the channel between control applications and field devices~\cite{zhu2011taxonomy} by transmitting measurement and control signals. The desire to improve the efficiency and reliability of control applications and SCADA has led to the use of heterogeneous and non-proprietary ICT~\cite{liu2012intruders}. However, this heterogeneous and non-proprietary ICT increases the number of cyber-vulnerabilities, opening up a much wider scope of cyber-security concerns among utilities.

By exploiting cyber-vulnerabilities, malicious adversaries can launch cyber-attacks against control applications and SCADA, among which \textit{Coordinated Cyber-Attacks~(CCAs)} are considered highly threatening and difficult to defend against. This is because CCAs (i) possess higher disruptiveness by integrating resources from multiple attack entities, and (ii) present heterogeneous traits in cyber-space and the physical grid by hitting multiple targets to achieve the attack goal. Thus, and as opposed to regular (or independent) attacks, whose severity is limited by the power grid's redundancy, CCAs could inflict catastrophic consequences as exemplified by the famous cyber-attacks against the Ukrainian power grid (the ``BlackEnergy'' malware attack in 2015~\cite{lee2016analysis,sun2016coordinated}, and the ``Crash Override'' attack in 2016~\cite{crashoverride}). These massive cyber-attacks could trigger power outages, leaving thousands of consumers and facilities without electricity.


Intrusion Detection Systems~(IDSs) are necessary tools to protect control applications and SCADA against cyber-attacks. IDSs record and analyze cyber-traces from adversaries that breach into the grid's cyber-system to exploit vulnerabilities. If, after analyzing cyber-traces, the security of the grid appears to be compromised, then IDSs will generate alarms. In addition, some IDSs will also take actions to mitigate attacks' effect. While IDSs can detect regular attacks or individual components of CCAs, they suffer from false alarms, fail to identify CCAs, and cannot estimate the attack consequences on the grid. 

To identify CCAs and estimate attack consequences, recent works suggest to integrate intrusion data from IDSs with attack templates --\textit{attack templates} model cyber-attacks against control applications.  This integration results in a set of \textit{Correlation Indices}~(CIs) describing the temporal and/or spatial correlation of coordinated attacks.  
\vspace{-\dist}
\subsection{Related Works} 
Many CIs have been proposed in the literature; however, they differ in their principles, which we summarize below.

\subsubsection{CIs based on adversaries' cyber traces} Attack sequences of the same adversary have similar cyber-traces that can be identified as contributing to CCAs. IDSs use this detection principle to investigate the temporal correlation of intrusions in cyber-space. Anomaly matrices~\cite{ten2011anomaly} and Time Failure Propagation Graphs~\cite{sun2016coordinated} are proposed to relate intrusion time with intrusion actions. While capable of detecting CCAs at the cyber-space, CIs of this type fail to estimate the attack consequences on the grid.

\subsubsection{CIs based on cyber-physical dependence} Logic graphs describing the conditions (in sequence in the cyber-space) for a physical consequence to take place can be used to derive CIs~\cite{wang2015cyber}. The logic graphs can take forms of attack trees \cite{ten2007vulnerability}, attack graphs \cite{liu2010security}, and PetriNets \cite{chen2011petri}. Temporal correlation of attacks is derived not only in the cyber-space but also in the physical power grid (see Fig.~2 in \cite{jie2015risk} for an example). However, constructing these logic graphs requires great computational effort due to the large number of cyber and physical components.

\subsubsection{CIs based on attack goals on the physical grid} Adversaries' goals described with reliability metrics or in terms of the criticality of a certain target are used to derive CIs. For example, in \cite{sun2016coordinated}, substations are attack targets and their criticality is first ranked. In \cite{ten2016cyber}, the attack goal is modeled as causing an insufficient power transfer. The work takes a numerical approach by disconnecting a set of substations at one time and running power flow. The substations in the set are identified as correlated if the power flow is divergent.

Given the great size of power grids, the combined deployment of CIs based on cyber-traces and attack goals promises better computation performance and higher accuracy than the CIs based on cyber-physical dependence.  The existing CIs based on attack goals, however, are limited to a few goals achieved by corrupting control commands. Other cyber-attacks, such as measurement attacks, present much higher threats in coordination (as a rich body of literature has shown their impact in electricity markets and security constrained power flows \cite{choi2013ramp, xie2011integrity, ye2016transmission}). This is because measurement attacks are (i)~difficult to detect by hiding in measurement signals and deceiving through control applications, and (ii)~capable of inflicting disastrous consequences by coordinating attacks against multiple grid components.
\vspace{-\dist}
\subsection{Our Work} 
This paper proposes a method to derive Correlation Indices based on attack goals for the following attack template: measurement attacks against Security Constrained Economic Dispatch~(SCED). In particular, we make the following contributions.
\begin{enumerate}
\item An analytical method to derive CIs. We formulate the attack template as a bilevel mix-integer optimization program. This problem is challenging due to its non-convex and combinatorial nature. To address these challenges, we propose an algorithm that computes the CIs based on attack goals.
\item A collection of set-theoretic properties for the CIs. These properties relate attack goals to the targets of CCAs. 
\item Defense strategies against CCAs, a metric of defense effectiveness, and the application of CIs to identify CCAs.
\end{enumerate}
Though we present our method to derive CIs for SCED, we emphasize that our method can be extended to other static control applications. 

The rest of the paper is organized as follows. Section 2 reviews the concepts of static applications and attack templates. The mathematical models of SCED and the attack template in bilevel form are presented in Sections 3 and 4, respectively. The CIs are derived in Section 5. Section 6 describes the CIs' properties, defense strategies, the metric of defense effectiveness, and the application of CIs to identify CCAs. In Section 7, the CIs are demonstrated with numerical experiments. Finally, Section 8 concludes the paper.

%% file: Background.tex
\section{Background}
In this section, we review the concepts of static applications and attack templates.
\vspace{-\dist} 
\subsection{Static Control Applications} 
\textit{Static Applications} are control loops designed to monitor, supervise, and control the grid's operating point --\ie they ignore the dynamics and work with the grid at a quasi-steady state. These applications can be automated or executed by a human operator. Examples include Security Constrained Economic Dispatch~(SCED), Optimal Reactive Power Support, and dispatch in Electricity Markets.

Fig.~\ref{fig:Hierarchical} illustrates a schematic of a static application. These applications compute control commands by solving optimization algorithms or by allowing direct manipulation via a Human Machine Interface. In any case, the control commands are computed based on measurements collected at remote substations.  To verify the integrity of these measurements, most well-known applications implement state estimation and bad data detection. 
\begin{figure}[t]
\centering
\includegraphics[width= .5\textwidth, height = 2.70 in]{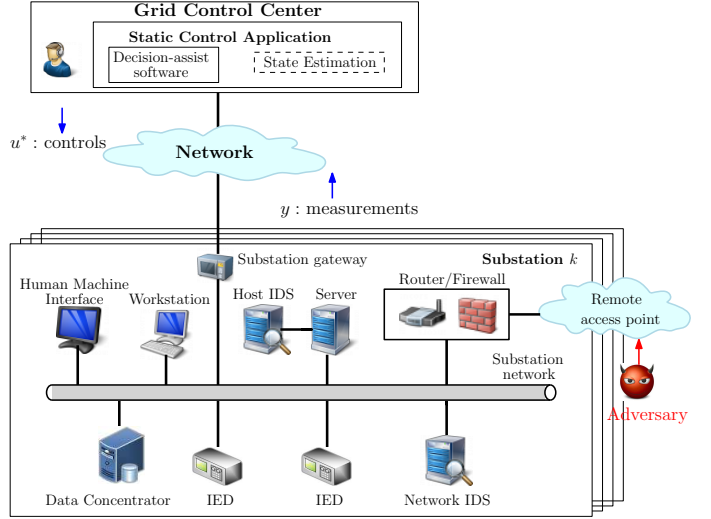}
\caption{Static Control Application. At substation $s_k$, we illustrate its Information and Communication Technology, including Intrusion Detection Systems~(IDSs).} 
\vspace{-\dist}
\label{fig:Hierarchical}
\vspace{-\dist}
\end{figure} 
\vspace{-\dist} 
\subsection{Attack Templates}
\textit{Attack templates} describe models of cyber-attacks on control applications. We consider two basic cyber-attacks: control and measurement attacks. In control attacks, the adversary modifies control commands directly~(Fig.~\ref{fig:ctrl}). In measurement attacks, the adversary modifies control commands indirectly by corrupting measurements (Fig.~\ref{fig:mattk}).  Since remote substations collect the measurements and operate physical devices~(\eg circuit breakers or capacitors), we assume that control and/or measurement attacks are executed by hacking into remote substations.
\begin{figure}[t]
    \centering 
    \subfloat[][Control attack]{\includegraphics[width=0.25\textwidth, height=1.05in]{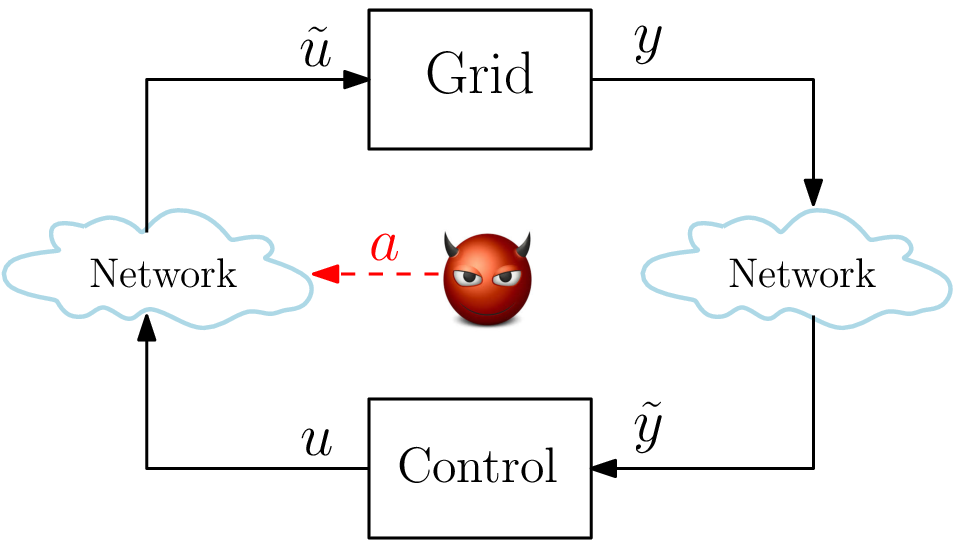}\label{fig:ctrl}}
    ~
    \subfloat[][Measurement attack]{\includegraphics[width=0.25\textwidth, height=1.05in]{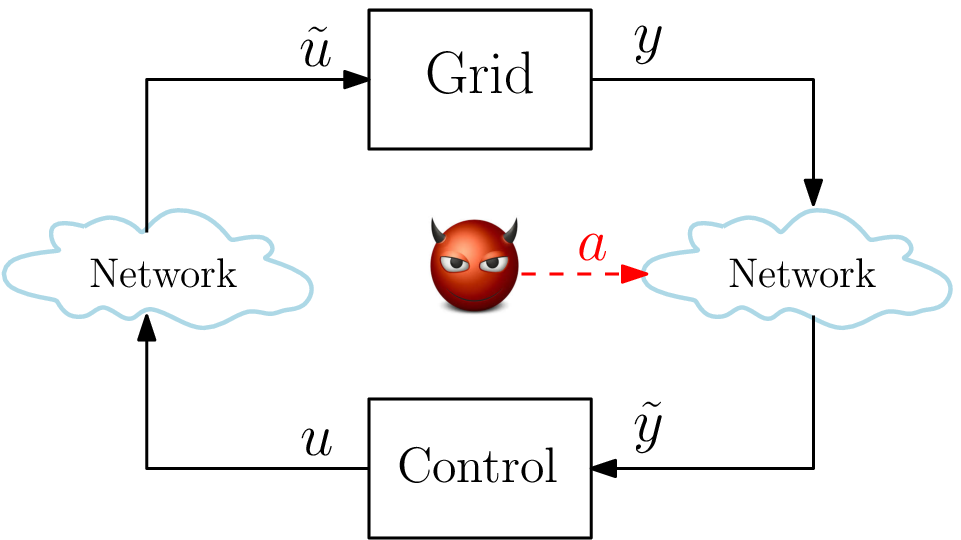}\label{fig:mattk}} 
    \caption{Control and measurement attacks. $u:$ control command, $y:$ measurements. $u \neq \tilde{u}$ ($y \neq \tilde{y}$) during the cyber-attack.}
    \vspace{-\dist}
    \label{fig:attacks}
\end{figure}

Attack templates have been used in the literature to determine the consequences of cyber-attacks, identify critical components of the grid, derive defense strategies, etc. For instance, by studying the attack template of measurement attacks on state estimation, several authors have proposed to stop the attacks by enhancing the screening methods of state estimation~\cite{liangreview}. This defense strategy, however, fails if the static application does not have state estimation, which is often true for real-time and contingency dispatch.

In this paper, we consider the following attack template: measurement attacks against SCED. This attack template describes an adversary with the following characteristics. 
\begin{enumerate}
\item The adversary knows the models of the power grid and SCED. 
\item The adversary can hack into the substations' ICT and inject falsified measurements to manipulate SCED.
\item The adversary can coordinate the attack against multiple substations over a large geographic area --\ie launch Coordinated Cyber-Attacks~(CCAs).
\end{enumerate}
We use the attack template to derive Correlation Indices~(CIs). These CIs describe a relation between the target substations and the attack goal (Fig.~\ref{fig:CIasGraph}).

\begin{remark}
The adversary's characteristics might be restrictive. However, they were selected for convenience of CIs' development and can be relaxed at the expense of more involved computations. For example, to relax the first characteristic, existing studies~\cite{tajer2017false,rahman2012false} developed stochastic methods to launch attacks with limited information. Other studies~\cite{liu2017false,tajer2011distributed} presented methods for estimating the power grid model with region-constrained information from multiple adversaries. The stochastic and estimation methods can easily be applied to extend the CIs' development method in future studies.
\end{remark}

\begin{figure}[t]
\centering
\includegraphics[width= .3\textwidth, height = 1.25 in]{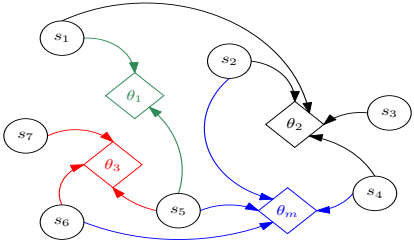}
\caption{Relation graphs between $n_s$ targets~($s_k$) and $m$ attack goals~($\theta_i$). For example, the targets associated to~$\theta_3$ are $\{s_5,s_6,s_{n_s}\}$, and to~$\theta_m$ are $\{s_2,s_4,s_5,s_6\}$}. 
\label{fig:CIasGraph}
\vspace{-\dist}
\end{figure} 

%% file: Modeling.tex
\section{Mathematical Models}
In this section, we describe the models of the power grid and Security Constrained Economic Dispatch~(SCED). 
\subsection{Mathematical Notation}
Throughout this paper, we use the following notation. Let $\R$ and $\R_{\geq 0}$ (resp. $\R_{>0}$) denote the set of real numbers and non-negative (resp. positive) real numbers. For $n>1$, $I_n$ denotes the $n$-dimensional identity matrix. $\mathbf{1}$ and $\mathbf{0}$ denote, respectively, the vectors (or matrices) with all components equal to one and zero. Given a finite set~$V$, we let $|V|$ denote its cardinality, \ie the number of elements of $V$, and $2^{V}$ the power set of $V$, \ie the set of all subsets of $V$. 

For a matrix $A \in \R^{n \times m}$, $[A]_i$ and $[A]_{ij}$ denote its $i$th row and its $(i,j)$th element. Given a vector $x \in \R^n$, $x_i$ denotes the $i$th element, $\diag(x)$ the diagonal matrix of $x$, and $||x||_0$ the zero norm of $x$, \ie the number of non-zero elements of~$x$. We let $||x||_{\infty}$ denote the infinity norm defined as $||x||_{\infty}:= \max\{|x_i|\}$. For two vectors $x,y\in \R^n$, $x \circ y = z \in R^n$ denotes the Hadamard or element-wise product, \ie $z_i = x_i y_i$, and $x \preceq y$ denotes the element-wise inequality, \ie $x_i \leq y_i$.
\vspace{-\dist}
\subsection{Power Grid Modeling} 
We model the power grid as the graph~$G = (V,E)$, where $V$ and $E \subset V \times V$ are the sets of $n := |V|$ buses and $m:=|E|$ transmission lines. To each bus $i \in V$, we associate the generation~$P_{g,i} \in \R_{\geq 0}$, and the demand~$P_{d,i} \in \R_{\geq 0}$; to each transmission line $e:=(i,j) \in E$, connecting buses $i,j \in V$, we associate the power flow $P_{f,e} \in \R$. In vector form, the generation, demand, and power flows are respectively $P_g = [P_{g,1}, \hdots, P_{g,n}]^\top$, $P_d = [P_{d,1}, \hdots, P_{d,n}]^\top$, and $P_f =[P_{f,1}, \hdots, P_{f,m}]^\top$. 

In addition, we assume the grid has a set of $n_s$ substations, \ie~$S  = \{s_1,s_2,\hdots,s_{n_s}\}$. At substation $s_k$, we represent the grid within its service area as the sub-graph $G_{s_k} = (V_{s_k},E_{s_k})$ with the following properties.
\begin{enumerate}
\item Substation service areas compose the entire power grid, \ie $\cup_{s_k \in S} G_{s_k} = G$. 
\item Substation service areas may overlap, \ie for some $s_k,s_l \in S$, we may have $G_{s_k} \cap G_{s_l} \neq \emptyset$, but the overlapped areas do not have buses with generation.
\item Each substation collects demand measurements, denoted as~$\tilde{P}_d \in \R^n$, within its service area. 
\end{enumerate} 
\vspace{-\dist}
\subsection{Security Constrained Economic Dispatch} 
We consider a SCED problem that computes a new generation profile~$P_g^*$ based on demand measurements~$\tilde{P}_d$. 

The SCED problem is formulated based on the power flow equations. The power flow equations are the mathematical model to plan, operate, and analyze the power grid. They describe how generation and demand balance, and how active and reactive power flow through the grid. 

For large-scale power grids, however, the coupled active and reactive power flow model might become computationally expensive and even unfeasible. Thus, a decoupled~(DC) power flow might be the only viable alternative to solve large-scale problems. DC power flow is simpler and more robust due to sparsity and linearity, but it is only accurate close to the operating point~\cite{dorfler2013novel}. We refer the interested reader to~\cite{PJM} and \cite{NE} for more information on how utilities use DC power flow.

We formulate SCED (based on DC power flow) as a convex optimization problem that minimizes the total generation cost~\eqref{eq:SCEDGenCost} subject to the following security constraints: generation-demand balance~\eqref{eq:SCEDPowerBalance}, operation limits of the generators~\eqref{eq:SCEDGensLimits}, and transmission limits on power flows~\eqref{eq:SCEDFlowLimits}, \ie
\begin{subequations} \label{eq:SCED}
\begin{align}
\min_{P_g} \quad & \frac{1}{2} P_g^\top C_2 P_g+ c_1^\top P_g +c_0, \label{eq:SCEDGenCost} \\
\text{s.t} \quad & \mathbf{1}^\top P_{g} - \mathbf{1}^\top \tilde{P}_{d} = 0, \label{eq:SCEDPowerBalance} \\
&  P_g \in [\mathbf{0} ,\bar{P}_g], \label{eq:SCEDGensLimits} \\
&  \underbrace{F (P_g - \tilde{P}_d)}_{=:P_f} \in [-\bar{P}_f,\bar{P}_f], \label{eq:SCEDFlowLimits}
\end{align}
\end{subequations}
where $c_2, c_1, c_0 \in \R_{\geq 0}^{n}$ are the cost coefficients for generation, $C_2 = \diag(c_2)$, $\bar{P}_g \in \R_{\geq 0}^n$ is the rated power from generators, $\bar{P}_f \in \R^m_{\geq 0}$ is the thermal capacity of transmission lines, and $F$ is the generator shift matrix.
\vspace{-\dist}

%% file: AttackTemplate.tex
\section{Attack Template}
In this section, we describe the attack template in bilevel form. The attack template models measurement attacks against SCED. We also describe the attack goal and constraints.
\vspace{-\dist}
\subsection{Measurement Attacks}
Let $a \in \R^n$ denote the attack signal. The adversary fabricates $a$ to corrupt measurements of the demand as follows
\begin{align} \label{eq:measA}
\tilde{P}_d(a) = P_d+a.
\end{align}
We assume the adversary injects $a$ by hacking into substations and altering measurements at the data concentrator (or at a communication link via a man-in-the-middle attack). Thus, in the rest of the paper, we refer the target data concentrator and ICT within the substation as the target substation.
\vspace{-\dist}
\subsection{The Attack Goal}
Using the corrupted measurements~\eqref{eq:measA}, the adversary has the following attack goal: to manipulate SCED and increase the power flow on a single target line~$e \in E$, which occurs at
\begin{align} \label{eq:AdversaryGoal}
|P_{f,e}(a)| = |[F]_e(P_g^*(a)-P_d)| \geq (1+\tau)|P_{f,e}(0)|, 
\end{align} 
where $P_{f,e}(a) \in \R$ (resp. $P_{f,e}(0) \in \R$) denotes the power flow on $e$ after (resp. before) the attack, $P_g^*(a) \in \R^n$ denotes the new (after the attack) generation profile, and $\tau \in (0, \bar{\tau}] \subseteq \R_{>0} $ quantifies the \textit{flow increase}. 

We use the notation $\eT \in E \times (0, \bar{\tau}]$ to describe attack goals satisfying~\eqref{eq:AdversaryGoal}. Since $\tau \in (0, \bar{\tau}]$, we can have (in theory) an infinite number of attack goals. In practice, however, we study a finite number of attack goals~$\tau$. For example, the attack goal $\tau$ that will cause congestion (relating to economic loss), overloading (increasing long-term capital cost by accelerating asset depreciation, increasing losses), and loss of transmission lines (under very stressful operating condition). Thus, in the worst case scenario, we assume the adversary maximizes the flow increase~$\tau$.
\vspace{-\dist}
\subsection{Attack Constraints}
The attack might be constrained due to the following.
\begin{enumerate}
\item State estimation and bad data detection.
\item Corruptible measurements and defense at substations.
\item Attack resources.
\end{enumerate}

Since SCED has state estimation, the adversary must design the attack signal $a$ to bypass bad data detection. Other applications, however, might not have state estimation, and hence the attack signal $a$ can take any (realistic) value. In any case, we write this constraint as  $||a||_{\infty} \leq \bar{a}$ where $\bar{a} > 0$. We can use $\bar{a}$ as a design parameter to model different attack scenarios.

If the defender protects substation~$s_k \in S$, then the adversary cannot corrupt measurements at $s_k$; otherwise, the adversary can corrupt all the measurements. We write this constraint as
\begin{align} \label{eq:AttackRegion}
a_i \in \delta_{s_k}[-\bar{a},\bar{a}], \ \forall i \in V_{s_k}, \forall s_k \in S,~\delta_{s_k} \in \{0,1\}, 
\end{align}
where $\delta_{s_k} = 1$ if the adversary attacks~$s_k$, and $\delta_{s_k} = 0$ if not. The vector $\delta(e, \tau) = [\delta_{s_1},\delta_{s_2},\hdots,\delta_{s_{n_s}}]^\top$ describes safe and target substations during CCAs with an attack goal $(e,\tau)$.

If the adversary has limited resources, then (s)he must limit the number of target substations. We write this constraint as
\begin{align} \label{eq:MinAttackedSubs}
||\delta (e,\tau)||_0 \leq \kappa,
\end{align} 
where $\kappa \in  \{1,2,\hdots,n_s\}$ denotes the maximum number of target substations. In the worst case scenario, the adversary minimizes $\kappa$.

\begin{remark}
Note that in the worst case scenario the adversary faces two conflicting objectives: maximize $\tau$ and minimize $\kappa$. The interaction $\tau - \kappa$ generates a Pareto-like behavior between aimed flow increase~($\tau$) and the number of target substations~($\kappa$).
\end{remark}
 \vspace{-\dist}
\subsection{Attack Template in Bilevel Form}
We use bilevel optimization to model the attack template, describing the worst case scenario of measurement attacks against SCED. Since \textit{bilevel optimization} models decision making among agents~\cite{luo1996mathematical} (\eg adversary vs defender), researchers have used it to study cyber-attacks~\cite{yuan2011modeling,liang2016vulnerability}); or physical attacks~\cite{arroyo2005solution} to power grids.

We write the attack template in bilevel form as follows:
\begin{align}
\max_{\tau,\kappa,\delta,a} \quad & \tau-\kappa,  \nonumber \\
\text{s.t.} \quad & \text{Eqs. } \eqref{eq:AdversaryGoal} - \eqref{eq:MinAttackedSubs}, \label{eq:AttackLevel}
\end{align}
where $P_g^*(a)$ denotes the optimal solution of the SCED optimization algorithm, parametrized by the attack signal~$a$, \ie
\begin{align}
P_g^*(a) \in \text{arg}\min_{P_g} \quad & \frac{1}{2}P_g^\top C_2 P_g + c_1^\top P_g + c_0,  \nonumber \\
\text{s.t} \quad & \mathbf{1}^\top P_g-\mathbf{1}^\top (P_d+a) = 0, \label{eq:SCEDAttackComp} \\
&  A_0 P_g + A_1a - b \preceq 0. \nonumber
\end{align}
with
\begin{align*}
A_0 := \begin{bmatrix}
-I_{n} \\
I_{n} \\
F \\
-F 
\end{bmatrix}, \quad A_1 = \begin{bmatrix}
\mathbf{0} \\
\mathbf{0} \\
-F \\
F
\end{bmatrix}, \quad b = \begin{bmatrix}
\mathbf{0} \\
\bar{P}_g \\
\bar{P}_f \\
-\bar{P}_f
\end{bmatrix} - A_1 P_d.
\end{align*}

In the above, the upper level problem~\eqref{eq:AttackLevel} models the attack goal and constraints, while the lower level problem~\eqref{eq:SCEDAttackComp} models the SCED manipulated through corrupted measurements~($a$).

The optimal solution of the bilevel form~$(\tau^*,\kappa^*,\delta^*,a^*,P_g^*)$, if it exists, describes an adversary that targets the least number of substations ($\kappa^*$ and $\delta^*$) and maximizes the flow increase~($\tau^*$) on the single line~$e \in E$. 

The bilevel form~\eqref{eq:AttackLevel}-\eqref{eq:SCEDAttackComp} depends on several parameters, including the power grid parameters, the SCED parameters, and the maximum value for the attack signal~$\bar{a}$. Thus, a defender, using the attack template, can select the parameters to study different scenarios.

\begin{remark}
By defining the corresponding attack goal, constraints, and control algorithm, we can model measurement attacks against other static applications, using the attack template in bilevel form. In addition, we can model control attacks using the upper level problem~\eqref{eq:AttackLevel}.  
\end{remark}
\vspace{-\dist}

%% file: DerivingCIs.tex
\section{Deriving the Correlation Indices}
In this section, we derive the key concepts, Correlation Indices~(CIs) and security index. We obtain the indices by transforming the attack template in bilevel form into a Mathematical Program with Equilibrium Constraints~(MPEC) and addressing its mathematical challenges.
\vspace{-\dist}
\subsection{Attack Template in Mathematical Programming Form} 
Since the lower level problem~\eqref{eq:SCEDAttackComp} is strictly convex on $P_g$ for a fixed $a$, its Karush-Kuhn-Tucker~(KKT) conditions are necessary and sufficient for optimality~\cite{boyd2004convex}. So, we can write the bilevel form \eqref{eq:AttackLevel}-\eqref{eq:SCEDAttackComp} as a Mathematical Program with Equilibrium Constraints~(MPEC) (\ie a single-level optimization problem~\cite{dempe2002foundations}), by replacing~\eqref{eq:SCEDAttackComp} with the KKT conditions. This yields
\begin{subequations} \label{eq:KKTform}
\begin{align}
\max_{\tau, \kappa, \delta, a, P_g^*} \quad & \tau-\kappa ,   \\
\text{s.t.} \quad & \text{Eqs. } \eqref{eq:AdversaryGoal} - \eqref{eq:MinAttackedSubs}, \\
&\mathbf{1}^\top P_g^*-\mathbf{1}^\top( a + P_d) = 0, \label{eq:KKTformbalance} \\
&C_2 P_g^*+ c_1 - \mathbf{1} \nu^* +A_0^\top \lambda^*=0, \label{eq:KKTformgradient} \\
&A_0 P_g^* + A_1 a -b \preceq 0, \label{eq:KKTformineq} \\
&\lambda^* \succeq 0, \label{eq:KKTformineqdual} \\
&\lambda^* \circ (A_0 P_g^* + A_1 a -b)= 0. \label{eq:KKTformcompslackness}
\end{align}
\end{subequations}
In the above, \eqref{eq:KKTformbalance}-\eqref{eq:KKTformcompslackness} are the KKT conditions of \eqref{eq:SCEDAttackComp} and $\nu^*$ (resp. $\lambda^*$) denotes the Lagrange multiplier of the equality (resp. inequality) constraint of~\eqref{eq:SCEDAttackComp}.
\vspace{-\dist}
\subsection{Mathematical Challenges} 
The MPEC~\eqref{eq:KKTform} is a challenging problem. Its properties are far more complex than the properties of traditional mathematical programming problems, making the standard nonlinear programming approach inapplicable~\cite{dempe2002foundations}. These challenges arise because the MPEC~\eqref{eq:KKTform} is nonconvex, is non-differentiable, and has two conflicting objectives.

The complementary slackness constraint~\eqref{eq:KKTformcompslackness} makes the MPEC~\eqref{eq:KKTform} nonconvex.  To address this challenge, we linearize \eqref{eq:KKTformcompslackness} using the Big M method~\cite{dempe2002foundations}. Let $M>0$ be a sufficiently large constant, then \eqref{eq:KKTformcompslackness} is equivalent to 
\begin{align}
\lambda^* &\preceq M(\mathbf{1}-\omega), \quad -(A_0P^*_g + A_1 a-b)\preceq M \omega, \label{eq:Linearcompslackness}
\end{align}
where $\omega \in \{0,1\}^{2(n+m)}$  is a binary decision variable.

The attack goal constraint~\eqref{eq:AdversaryGoal} makes the MPEC~\eqref{eq:KKTform} non-differentiable. To address this challenge, we proceed as follows. Since the flow on $e$ before the attack~$P_{f,e}(0)$ can be computed using~\eqref{eq:SCED}, the attack goal constraint \eqref{eq:AdversaryGoal} can be written as
\begin{align} \label{eq:EquivChangeFlow}
\begin{cases}
[F]_{e}(P_g^*(a)-P_d)\geq (1+\tau)\ P_{f,e}(0), \ \text{if } P_{f,e}(0) \geq 0, \\
[F]_{e}(P_g^*(a)-P_d)\leq (1+\tau)\ P_{f,e}(0), \ \text{if } P_{f,e}(0) < 0.
\end{cases}
\end{align}

The MPEC~\eqref{eq:KKTform} has two conflicting objectives, \ie $\max~\tau - \kappa$. To address this challenge, we minimize $\kappa$ (\ie the number of target substations) and let $\tau \geq \tilde{\tau}$ where $\tilde{\tau}$ is a predefined flow increase. We can attach semantics to $\tilde{\tau}$, \eg the~$(\tilde{\tau})$ that triggers the line's protection.

The proposed solutions for the challenges transform the MPEC~\eqref{eq:KKTform} into the following mixed-integer linear programming problem
\begin{align} \label{eq:MinCard}
\min_{\tau, \kappa,\delta,a,P_g^*} \quad & \kappa, \nonumber \\
\text{s.t.} \quad  &\tau \geq \tilde{\tau},  \\
&\text{Eqs. } \eqref{eq:AttackRegion},\eqref{eq:MinAttackedSubs},\text{\eqref{eq:KKTformbalance}-\eqref{eq:KKTformineqdual}, \eqref{eq:Linearcompslackness}}\text{ and } \eqref{eq:EquivChangeFlow}. \nonumber
\end{align}
\vspace{-\dist}
\subsection{Algorithm: Deriving the CIs} 
The optimal solutions $\kappa^*$ and $\delta^*\eT$ of \eqref{eq:MinCard} denote, respectively, the security index and the Correlation Index~(CI) for the attack goal~$\eT$. The \textit{security index}~$\kappa^*$ determines the least number of target substations to increase the flow~$(\tilde{\tau})$ on line $e$, while the \textit{Correlation Index}~$\delta^* \eT$ describes which target substations. This CI represents a strongly correlated CCA since it relates the least number of target substations with the attack goal $\eT$.

Though the security index $\kappa^*$ is unique, the CI might not be. Other CCAs attacking $\kappa^*$ substations might also increase the flow~$(\tilde{\tau})$ on line $e$ --this is a consequence of the combinatorial nature of \eqref{eq:MinCard}. All the CIs, however, are feasible solutions of \eqref{eq:MinCard} with $\kappa = \kappa^*$, which we use to develop the following algorithm.

Given the attack goal~$\eT$, Algorithm~\ref{alg:ImpactiveAttacks} computes the security index first, and then the CIs by exploring which of the ${\kappa^*}\choose{n_s}$ combinations of target substations are feasible solutions of \eqref{eq:MinCard} with $\kappa = \kappa^*$.

\begin{algorithm}[t]
\caption{Deriving the Security Index and CIs}\label{alg:ImpactiveAttacks}
\begin{algorithmic}[1]
\State $(\kappa^*,\text{CIs}) \gets$ CorrelationIndices($e,\tilde{\tau}$)
\Procedure{CorrelationIndices}{$e,\tilde{\tau}$}  
\State $\text{CIs} \gets \{\emptyset\}$ 
\State Compute $\kappa^*$ by solving~$\eqref{eq:MinCard}$
\For{$j = 1$ to ${\kappa^*}\choose{n_s}$}
\If{$\delta(j)$ is \textit{feasible} (of~\eqref{eq:MinCard} with $\kappa = \kappa^*$)} 
\State $\text{CIs} = \{\text{CIs},\delta(j)\}$
\EndIf
\EndFor
\State \textbf{return} $(\kappa^*,\text{CIs})$
\EndProcedure
\end{algorithmic}
\end{algorithm}
\vspace{-\dist}
\subsection{Limitations} Our method has a limitation, namely the computation performance of Algorithm~\ref{alg:ImpactiveAttacks}, which we discuss next.

Algorithm~\ref{alg:ImpactiveAttacks} only promises local optimal solutions in finite time. This is because the mixed-integer linear problem~\eqref{eq:MinCard} and the ${{\kappa^*}\choose{n_s}} - 1$ feasibility problems are in general NP-hard. Given that there is only a few substations in a power grid, the computation time of the proposed algorithm is unlikely to be a problem. However, in case of abrupt changes occurring in the power grid, CIs will need to be updated at run-time and an algorithm providing theoretic bounds of convergence must be sought after. These tasks are out of the scope of this paper, but they will be part of our future work.  
\vspace{-\dist}

%% file: CIsProperties.tex
\section{Applying the CIs to Protect against CCAs}
In this section, we describe properties of the Correlation Indices~(CIs) using a set theoretic approach. These properties allows us to derive defense implications and applications in IDS.
\vspace{-\dist}
\subsection{CIs' Properties}
Let $\Sa \in 2^S$ describe the set of target substations during a CCA. If the CCA is effective, \ie if the CCA increases the flow~($\tilde{\tau}$) on line $e \in E$, we use the notation $\Sa \to \eT$; otherwise we use $\Sa \not \to \eT$. We collect all effective CCAs in the set \[\CSa := \{\Sa~|~\Sa \to (e,\tilde{\tau}) \} \subset 2^S.\] The next proposition shows that if the CCA~$\Sp$ fails to increase the flow~($\tilde{\tau}$) on line $e$, then all subordinated attacks~$\Sa \subset \Sp$ also fail to increase the flow on $e$. 
\begin{proposition} \label{prop:2}
(\textit{Subordinated CCAs.}) If $\Sp \not \in \CSa$, then $\Sa \not \in \CSa$ for any $\Sa \subset \Sp$.
\end{proposition}
\begin{proof}
Appendix.
\end{proof} 
\vspace{-\dist}
\begin{definition}
Let $\delta^* \eT$ denote a feasible solution of Algorithm \ref{alg:ImpactiveAttacks} (\ie $||\delta^* \eT ||_0 = \kappa^*$). A CI, denoted as $\CI$, is a strongly correlated CCA that extracts target substations from $\delta^* \eT$ as follows
\begin{equation} \label{eq:CI}
\CI :=\{s_k \in S~|~\delta_k^* \eT \neq 0 \},
\end{equation}
and reaches the goal $\eT$, \ie $\CI \in  \CSa$.
\end{definition}
A CI has \textit{minimal cardinality}, \ie the CI is effective by attacking the least number of substations, which we state next. 
\begin{proposition} \label{prop:1}
(\textit{Minimal Cardinality.}) Let $\CI \in \CCI$ be a CI. Then $\Sp \not \to \eT$ for any $\Sp \in 2^S$ satisfying $|\Sp| < \kappa^*$.
\end{proposition}
\begin{proof}
Appendix.
\end{proof}
Note that Proposition~\ref{prop:2} and \ref{prop:1} guarantee security against subordinated CCAs of the CI~$\CI$.

CIs are not unique since there might be another CCA~$\Sp$ satisfying $|\Sp| = \kappa^*$ and $\Sp \in \CSa$. We collect all CIs in the set \[\CCI:= \{ \CI~|~\CI \txt{ is a CI} \} \subseteq \CSa.\]

The following lemma states that any CCA containing the CI~$\CI \in \CCI$ can increase the flow on line~$e \in E$. 
\begin{lemma} \label{lemma:1}
Let $\CI \in \CCI$ denote a CI. Suppose there exists another set of target substations $\Sp \in 2^S$, then $(\CI \cup \Sp) \in \CSa$.
\end{lemma}
\begin{proof}
Appendix.
\end{proof}
Lemma~\ref{lemma:1} implies that CCAs might be targeting multiple lines, which we generalize in the next theorem.
\begin{theorem} \label{theorem:1}
(\textit{Multiple Targets})~Let $J$ be an arbitrary index set, and let $\CI \in  \mathcal{S}^*(e_j, \tilde{\tau}_j)$ be CIs (for different lines) for all $j \in J$. Suppose there exists a CCA targeting substations $S^* \in 2^S$ such that $\CI \subset S^*$ for all $j \in J$. Then the CCA $S^*$ can increase the flow~($\tilde{\tau}_j$) on any line from the set $\{e_j\}_{j \in J}$, \ie $S^* \in \mathcal{S}(e_j,\tilde{\tau}_j) $ for all $j \in J$.
\end{theorem} 
\begin{proof}
Appendix.
\end{proof}

The CIs' properties described allow us to study what happens if we protect the measurements on a substation, which we state in the next theorem.
\begin{theorem} \label{theorem:2}
(\textit{Defense at a Single Substation})~Let $J$ be an arbitrary index set, and let $\CI \in  \mathcal{S}^*(e_j, \tilde{\tau}_j)$ be CIs (targeting different lines) for all $j \in J$. Moreover, suppose the collection of target substations $\{\CI\}_{j \in J}$ satisfies $\cap_{j \in J} \CI \neq \emptyset$. If the grid's defender protects measurements at substation $s_k^* \in \cap_{j \in J} \CI$, then one of the following occurs (for all $j \in J$):
\begin{enumerate}
\item[(i)] If $\CI$ is the unique CI that increases the flow~($\tilde{\tau}_j$) on line $e_j$, then after $s_k^*$ is protected, the new CI $S^*_{\beta,j}$ (not necessarily unique) will require targeting more substations, \ie $\kappa^*_{\beta}:= |S^*_{\beta,j}| > |\CI|$.
\item[(ii)] If $\CI$ is not the only CI, then, after $s_k^*$ is protected, the following might occur: (a) if $s_k^*$ is common to all CIs $\CI \in \CCI$, then the conclusion in (i) applies; or (b) if $\{s_k^*\} \cap \CI = \emptyset$ for some $\CI \in \CCI$, then the new collection of CIs, denoted as $\mathcal{S}_\beta^*(e_j,\tilde{\tau}_j)$, satisfies $\mathcal{S}^*_\beta(e_j,\tilde{\tau}_j)  \subset \mathcal{S}^*_\alpha(e_j,\tilde{\tau}_j)$.
\item[(iii)] The attack is infeasible, \ie $\mathcal{S}^*_\beta(e_j,\tilde{\tau}_j) \equiv \{\emptyset \}$.
\end{enumerate}
\end{theorem}
\begin{proof}
Appendix.
\end{proof}
Theorem~\ref{theorem:2} suggests that protecting a substation will pivot the CI (from one set to a different set). It also suggests that protecting substation~$s_k \in S$ becomes more critical if $s_k$ is related to more target lines. We discuss more defense implications next.
\vspace{-\dist}

%% file: DefenseImplications.tex
\subsection{Defense Implications}
In this subsection, we derive the best defense for a line~$e \in E$, the best defense against CIs, and the best defense for substations based on the CIs' properties.

\subsubsection{The Best Defense for a Line~$e \in E$} 
Suppose the CCA $\Sa$ increases the flow~($\tilde{\tau}$) on line $e \in E$, \ie $\Sa \in \CSa$. Then, the best defense for $e$ is to protect the minimal set of substations~$\Db \subseteq \Sa$ that renders the new CCA $\Sp:= \Sa \setminus \Db$ ineffective, \ie $\Sp \not \in \CSa$. If $\Sa$ is a CI (\ie $\Sa \in \CCI$), then $\Db = \{s_k\}$ with $s_k \in \Sa$, \ie protecting any substation from $\Sa$ renders the new CCA ineffective.

On the other hand, CCAs might not remain static; that is, the adversary might switch between a CI~$\CI \in \CCI$ and a CCA~$\Sa \in \CSa$ to identify vulnerabilities and hide from detection. However, if the attacks ($\CI$ and $\Sa$) have common substations, \ie if $\CI \cap \Sa \neq \emptyset$, then the best defense for line~$e$ is to protect substations satisfying $s_k^* \in \CI \cap \Sa$.

\subsubsection{The Best Defense against CIs}
Suppose $\kappa^*$ denotes the security index for the attack goal $(e,\tilde{\tau})$. Then, the best defense against CIs (\ie $\CCI$) is to protect the minimal set of substations~$\Db$ that renders the new security index~$\kappa^*_{\beta}$ greater than $\kappa^*$. Thus, the adversary is required to attack more substations after~$\Db$ is protected. Note that Theorem~\ref{theorem:2} (i) and (ii-a) describe two special cases of this defense, \ie when $\Db \equiv \{s_k^*\}$.

\subsubsection{Metrics of Defense Effectiveness} 
We describe the metric used to compare defense at substations and to identify the best defense strategy. Suppose $\kappa^*\eT$ denotes the security index for the attack goal $\eT$. The security index measures the likelihood of a CCA since it is less likely to attack more substations than $\kappa^* \eT$. We define the \textit{average likelihood} to increase the flow~($\tilde{\tau}$) on all lines as \[\Ri = \frac{1}{m} \sum_{e \in E} \kappa^*\eT.\]

Using the average likelihood, we have the following definition.
\begin{definition} \label{def:Risk}
For a target flow increase~$\tilde{\tau}$, the defense effectiveness for substation $s_k$ can be estimated by calculating 
\begin{align*}
\DRb := \Rb - R(\tilde{\tau}),
\end{align*}
where $\Rb := \frac{1}{m} \sum_{e \in E} \kappa^*_{\beta,s_k} \eT$ denotes the average likelihood after protecting substation $s_k$.
\end{definition}

Theorem \ref{theorem:2} implies that $\DRb \geq 0$ for all $s_k \in S$, \ie after protecting substation~$s_k$, the number of target substations increases, while the average likelihood decreases. Thus, the best defense strategy is to protect substation $s_k^*$ such that $\Delta R_{\beta,s_k^*} \in \text{arg}\max_{s_k \in S} \{\Delta R_{\beta,s_k}\}$. Note that we derive the metric for a specific flow increase value~($\tilde{\tau}$), but we can always derive it for the case when $\tilde{\tau} \in (0,\bar{\tau}]$ is a free parameter for all target lines.
\begin{remark}
If we integrate $\kappa^*$ (\ie the likelihood of CCAs) and the likelihood of exploits at the cyber network, we can derive a \textit{risk} metric for CCAs at both the cyber and physical networks. This metric will be studied in our future work. 
\end{remark}
\vspace{-\dist}

%% file: Applications.tex
\subsection{Application: Identifying CCAs}
In this subsection, we briefly describe how CIs based on cyber-traces and CIs based on attack goals identify CCAs and estimate their possible consequences.

CIs based on cyber-traces identify in real time individual components of CCAs, that is, the set of suspected target substations. These CIs interpret intrusion data from sensors of Intrusion Detection Systems (or other security tools) installed at substations.

CIs based on attack goals estimate the possible consequences of CCAs. The grid's defender computes these CIs and stores them in a knowledge base. The grid's defender can update this knowledge base of CIs as needed.

Fig.~\ref{fig:schem} depicts a schematic of how the CIs work together. The CIs based on cyber-traces output the set of suspected target substations~$\Sa(t)$ (at some time~$t$) to the knowledge base of CIs. The knowledge base of CIs compares this set of suspected target substations with the CIs (and their mathematical properties) to estimate possible consequences~$\eT$. This approach is analogous to signature-based (also know as blacklist) detection techniques~\cite{zhu2005alert}. Nevertheless, the CIs (signatures) are derived based on the attack template instead of direct network knowledge. Thus, combining CIs with other direct-knowledge based approaches in IDS will significantly improve defense performance and allow taking immediate actions upon most harmful attacks. We will provide details of this approach in a different paper. 
\begin{figure}[t!]
\centering
\includegraphics[width= .45\textwidth, height = 0.5 in]{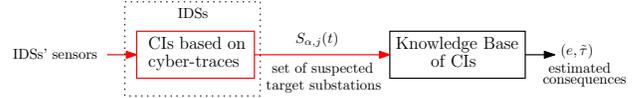}
\caption{Schematic: Identifying CCAs.}
\label{fig:schem}
\vspace{-\dist}
\end{figure}

\vspace{-\dist}

%% file: Experiments.tex
\section{Numerical Experiments}
In this section, we provide numerical experiments to illustrate the CIs, their properties, defense implications, and the metric of defense effectiveness. 

Fig.~\ref{fig:NewEngland39} shows the New England 39 bus  system used to model a power grid with $n_s = 6$ substations. We selected two target lines, $e = (2,25)$ and $e' = (16,21)$; and the flow increase $\tilde{\tau} \in \{2.5\%, 5\%,7.5\%,10\%\}$ to describe attack goals. The parameters used in our experiments were: $M = 500$, $\bar{a} = 0.1$, and the SCED base case data for the New England system taken from MATPOWER software package~\cite{zimmerman2011matpower}.  
\vspace{-\dist}
\begin{figure}[t!]
\centering
\includegraphics[width= .42\textwidth, height = 2.25 in]{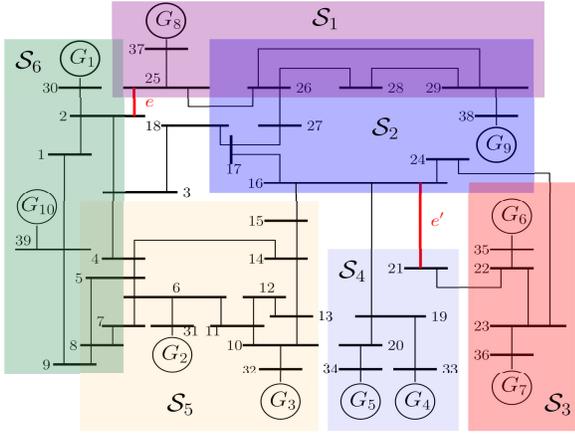}
\caption{New England 39 bus system.}
\label{fig:NewEngland39}
\vspace{-\dist}
\end{figure}
\subsection{Experiment 1. Computing the CIs}
In this experiment, we derived the CIs for the attack goals $(e,\tilde{\tau})$ and $(e',\tilde{\tau})$ using Algorithm~\ref{alg:ImpactiveAttacks}. We implemented Algorithm~\ref{alg:ImpactiveAttacks} using CVX (a package for solving convex and linear mixed integer programs~\cite{cvx}). Tables \ref{table:CIsfore} and \ref{table:CIsforeP} present the collection of CIs. We found that all attack goals have unique CIs but $(e,5 \%)$.
\vspace{-\dist}
\begin{table}[!h]
\caption{Correlation Indices for line $e=(2,25)$.}
\label{table:CIsfore}
\centering
\begin{tabular}{c c c c c c c} 
\hline
Attack goal  & CIs  & Security index \\  [0.5ex] 
\hline
 $\mathcal{S}^*_{\alpha}(e,2.5\%)$ & $\{2\}$  & $\kappa^* = 1$ \\
$\grhl{\mathcal{S}^*_{\beta}(e,2.5\%)}$ & \grhl{$\{5,6\}$} & \grhl{$\kappa^*_\beta = 2$}  \\ 
 $\mathcal{S}^*_{\alpha}(e,5\%)$ & $\{2,5,6\},\{4,5,6\}$  & $\kappa^* = 3$ \\
 $\grhl{\mathcal{S}^*_{\beta}(e,5\%)}$ & \grhl{$\{4,5,6\}$} & \grhl{$\kappa^*_\beta = 3$}  \\ 
 $\mathcal{S}^*_{\alpha}(e,7.5\%)$ & $\{1,2,5,6\}$  & $\kappa^* = 4$ \\
 $\grhl{\mathcal{S}^*_{\beta}(e,7.5\%)}$ & \grhl{$\{\emptyset\}$} & \grhl{$\kappa^*_\beta = 0$}  \\ 
 $\mathcal{S}^*_{\alpha}(e,10\%)$ & $\{S\}$  & $\kappa^* = 6$ \\
 $\grhl{\mathcal{S}^*_{\beta}(e,10\%)}$ & \grhl{$\{\emptyset\}$} & \grhl{$\kappa^*_\beta = 0$}  \\ 
 \hline
\end{tabular}

\begin{tablenotes}
\small
\item \textit{Note:} $\CCI$ (resp. $\CCIb$) denotes the CIs before (resp. after) protecting $s_k^* = 2$.  $\{\emptyset\}$ (or $\kappa^* = 0$) implies the attack is ineffective.
\end{tablenotes}
\end{table}

\begin{table}[!h]
\caption{Correlation Indices for line $e'=(21,24)$.}
\label{table:CIsforeP}
\centering
\begin{tabular}{c c c c c c c} 
\hline
Attack goal  & CIs  & Security index \\  [0.5ex] 
\hline
 $\mathcal{S}^*_{\alpha}(e',2.5\%)$ & $\{2\}$  & $\kappa^* = 1$ \\
$\grhl{\mathcal{S}^*_{\beta}(e',2.5\%)}$ & \grhl{$\{5,6\}$} & \grhl{$\kappa^*_\beta = 2$}  \\ 
 $\mathcal{S}^*_{\alpha}(e',5\%)$ & $\{4,5,6\}$  & $\kappa^* = 3$ \\
 $\grhl{\mathcal{S}^*_{\beta}(e',5\%)}$ & \grhl{$\{4,5,6\}$} & \grhl{$\kappa^*_\beta = 3$}  \\ 
 $\mathcal{S}^*_{\alpha}(e',7.5\%)$ & $\{\emptyset\}$  & $\kappa^* = 0$ \\
 $\grhl{\mathcal{S}^*_{\beta}(e',7.5\%)}$ & \grhl{$\{\emptyset\}$} & \grhl{$\kappa^*_\beta = 0$}  \\ 
 \hline
\end{tabular}

\begin{tablenotes}
\small
\item \textit{Note:}~$\mathcal{S}^*_\alpha(e',\tilde{\tau})$ (resp. $\mathcal{S}^*_\beta(e',\tilde{\tau})$) denotes the CIs before (resp. after) protecting $s_k^* = 2$. $\{\emptyset\}$ (or $\kappa^* = 0$) implies the attack is ineffective.
\end{tablenotes}
\end{table}
\subsection{Experiment 2. CIs Dependence on the parameter $\bar{a}$}
In this experiment, we studied the CIs' dependence on $\bar{a}$. Fig.~\ref{fig:depa} shows how the security index $\kappa^*$ changes as we increase the attack signal max value $\bar{a}$. We found that the security index decreases as $\bar{a}$ increases. This result imply that if the defender increases $\bar{a}$ in the attack template, the defense implications become more conservative.
\vspace{-\dist}

\begin{figure}[t!]
\centering
\includegraphics[width= .45\textwidth, height = 2 in]{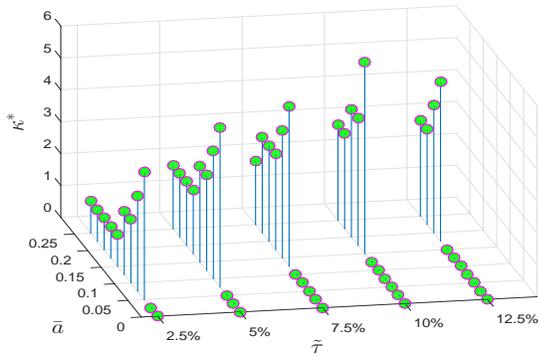}
\caption{CIs' dependence on $\bar{a}$. Target Line $e=(2,25)$}
\label{fig:depa}
\vspace{-\dist}
\end{figure}
\subsection{Experiment 3. Mathematical Properties of CIs}
We studied the mathematical properties of CIs and defense implications from Theorem~\ref{theorem:2}. Tables~\ref{table:CIsfore} and \ref{table:CIsforeP} show, respectively, the CIs for the attack goals $\eT$ and $(e', \tilde{\tau})$, before and after protecting substation~$s_k^* = s_2$. Before protecting substation~$s_k^* = s_2$, the CIs have the following defense implications: for the attack goal~$(e, 5\%)$, subordinated attacks of the CI $\CI = \{2,5,6\}$ are ineffective (Propositions~\ref{prop:2} and \ref{prop:1}); and the CCA $S^* = \{2,4,5,6\}$ can increase the flow $\tilde{\tau} = 5\%$ on both lines $e$ and $e'$ (Lemma~\ref{lemma:1} and Theorem~\ref{theorem:1}). And, after protecting substation~$s_k^* = 2$, the CIs have the following defense implications: for the attack goal $(e',2.5\%)$, the new security index satisfies $\kappa^*_\beta = 2 > 1 = \kappa^*$~(Theorem~\ref{theorem:2} (i)), $\mathcal{S}^*_\beta(e,5\%) \subset \mathcal{S}^*_\alpha(e,5\%)$~(Theorem~\ref{theorem:2} (ii-b)), and $\mathcal{S}^*_\beta(e,7.5\%) = \{\emptyset\}$~(Theorem~\ref{theorem:2}~(iii)). 
\vspace{-\dist}
\subsection{Experiment 4. The Metric of Defense Effectiveness.}
In this final experiment, we computed the metric of defense effectiveness~$\Delta R_{s_k}(\tilde{\tau})$ for all substations $s_k \in S$ and $\tilde{\tau} \in \{ 5\%,7.5\%\}$. Table~\ref{table:risk} presents the results. These results imply that the best defense is achieved by protecting substation~$s_k^*=s_2$.
\vspace{-\dist} 
\begin{table}[!t]
\caption{Metric of Defense Effectiveness.}
\label{table:risk}
\centering
\begin{tabular}{c c c c c c c} 
\hline
 Substation & \grhl{$s_1$} & $s_2$ & \grhl{$s_3$} & \grhl{$s_4$} & \grhl{$s_5$} & \grhl{$s_6$} \\  [0.5ex] 
 \hline
 $\Delta R_{s_k}(5\%)$ & \grhl{0.11} &  0.70 & \grhl{0.02} & \grhl{0.02} & \grhl{0.30} & \grhl{0.30} \\ 
 $\Delta R_{s_k}(7.5\%)$ & \grhl{0.13} &  0.54 & \grhl{0.02} & \grhl{0.07} & \grhl{0.2} & \grhl{0.2} \\
 \hline
\end{tabular}
\vspace{-\dist}
\end{table}

%% file: Proofs.tex
\begin{proof}
(Proposition~\ref{prop:1}) Assume $\CI \in \CCI$. We partition the CI in two arbitrary disjoint sets, \ie $\CI = \Sp \cup D$ satisfying $|\Sp| < \kappa^*:= |\CI|$. Since CIs are minimum cardinality CCAs, then $\Sp \not \in \CSa$, which proves the proposition. 
\end{proof}
\vspace{-\dist}
\begin{proof}
(Lemma~\ref{prop:2}) Suppose, to get a contradiction, $\Sa \in \CSa$; then  any super-set~$S^*$ of $\Sa$, \ie $\Sa \subseteq S^*$, is effective, \ie $S^* \in \CSa$. In particular, $\Sp \equiv S^*$ reaches the goal $\eT$, which contradicts $\Sp \not \in  \CSa$.  
\end{proof}
\vspace{-\dist}
\begin{proof}
(Lemma~\ref{lemma:1}) We prove the lemma by cases. (i) Suppose $\CI \in \CCI$ is a super-set of $\Sp$, \ie $\CI \supseteq  \Sp$. It follows that $\Sp \cup \CI \equiv \CI \in \CCI \subseteq \CSa$. Similarly, (ii) suppose that $\CI \subseteq \Sp$. Then, it follows that $ \Sp \cup \CI = \Sp \in \CSa$. Finally, (iii) suppose the CI $\CI$ is neither a subset or super-set of $\Sp$. Assume, to get a contradiction, that $(\Sp \cup \CI) \not \in \CSa$. Then, Proposition~\ref{prop:2} implies that any CCA $\Sa \subseteq (\Sp \cup \CI)$ does not reach the goal $\eT$. In particular, $\Sa \equiv \CI \subset (\Sp \cup \CI)$ does not reach $\eT$, \ie $\CI \not \in  \CSa$. This yields the contradiction.
\end{proof}
\vspace{-\dist}
\begin{proof}
(Theorem~\ref{theorem:1}) We partition the CCA $S^*$ in the union of two disjoint sets, \ie $S^* = \CI \cup (S^* \setminus \CI)$ for all $j \in J$. Then, by Lemma~\ref{lemma:1} we have $S^* \in \mathcal{S}_\alpha(e_j,\tilde{\tau}_j)$ for all $j \in J$, which proves the theorem.
\end{proof}
\vspace{-\dist}
\begin{proof}
(Theorem~\ref{theorem:2}) Assume the operator protects substation $s_k^*$. Moreover, assume that $\CI \in \CCI$ is unique and the attack remains feasible after defense. Suppose, to get a contradiction, that the new security index satisfy $\kappa_\beta^* = \kappa^*$. This implies that the new CI satisfies $\CIb \in \CCI$, which contradicts the fact that $\CCI = \{\CI\}$. This proves (i). On the other hand, suppose that $\CI \in \CI$ is not unique. Part (ii-a) can be proven using the same arguments as in (i). We prove part (ii-b) as follows. Define $\mathcal{S}: = \{\CI \in \CCI~|~{s_k^*} \cap \CI = \emptyset\}$. Then we partition the collection $\CCI$ as follows $\CCI = \mathcal{S} \cup (\CCI \setminus \mathcal{S})$. Proposition 1 implies that after defense $\CI \setminus \{s_k^*\} \not \in \CCIb$ for all $\CI \in (\CCI \setminus \mathcal{S})$. Thus, $\CCIb \equiv \mathcal{S}$, and therefore $\CCIb \subset \CCI$. This proves (ii-b). Finally, (c) follows trivially.  
\end{proof}
\vspace{-\dist}